\documentclass[journal]{IEEEtran}

\usepackage[utf8]{inputenc}
\usepackage{balance}
\usepackage{cite}
\usepackage{float}
\usepackage{multicol,multirow}
\usepackage{makecell,booktabs}
\usepackage{xurl}
\usepackage{hyperref}
\hypersetup{
    colorlinks=true,
    linkcolor=black,
    anchorcolor=black,
    citecolor=black,
    filecolor=black,
    menucolor=black,
    runcolor=black,
    urlcolor=black,
}
\usepackage{subfigure}
\usepackage{graphicx}
\usepackage{amsfonts}
\usepackage{amssymb}
\usepackage{amsmath,mathtools}
\usepackage{array}
\usepackage{longtable}
\usepackage[ruled,vlined]{algorithm2e}
\usepackage{color}
\usepackage{svg}
\usepackage{bbm}
\usepackage{bm}
\usepackage{comment}
\usepackage{forest}
\usepackage{etoolbox}
\usepackage{xspace}
\usepackage{tabularx}
\usepackage{amsthm}
\usepackage{cleveref}

\newtheorem{lemma}{Lemma}
\theoremstyle{remark}

\newcolumntype{L}[1]{>{\raggedright\let\newline\\\arraybackslash\hspace{0pt}}m{#1}}
\newcolumntype{C}[1]{>{\centering\let\newline\\\arraybackslash\hspace{0pt}}m{#1}}
\newcolumntype{R}[1]{>{\raggedleft\let\newline\\\arraybackslash\hspace{0pt}}m{#1}}
\newcolumntype{Y}{>{\raggedright\arraybackslash}X}
\newcolumntype{P}[1]{>{\raggedright\arraybackslash}p{#1}}

\begin{document}

\author{
    Hyun~Joong~Kim,~\IEEEmembership{Member,~IEEE,}
}

\title{Zero-Knowledge Verification of Transaction Guides for P2P Energy Trading in Distribution Networks}
\maketitle

\begin{abstract}
Peer-to-peer (P2P) energy trading requires network-aware coordination because transactions are physically realized through distribution networks. However, Sensitivity-based coordination causes a confidentiality–verifiability tradeoff, as network sensitivities may reveal vulnerable components while undisclosed sensitivities prevent participants from verifying utility-provided transaction guides. This paper proposes a zero-knowledge-proof-based method for verifying the computational integrity of network-constrained transaction guides with respect to committed private network data, without exposing network-sensitivity information. The guide defines admissible injection and withdrawal volumes derived from sign-decomposed sensitivity matrices while satisfying balance, voltage, line-flow, and optimality conditions. These conditions are encoded in an arithmetic circuit, represented as R1CS constraints and a quadratic arithmetic program, and verified using a bilinear pairing. Blockchain commitments bind the approved circuit, public inputs, statement identifiers, proof, and verification result for tamper-evident auditability. The proposed proof certifies correct guide computation from committed network data; the authenticity of the committed network data is handled through an explicit registration and attestation assumption. Case studies on a modified IEEE 33-bus system show satisfying network constraints post-clearing, rejection of public-input and witness-inconsistency attacks, and practical on-chain overhead, with an 806-byte proof.
\end{abstract}

\section{Introduction}
The increasing penetration of renewable and distributed energy resources (DERs) is turning passive distribution-level end-users into active prosumers, making peer-to-peer (P2P) energy trading a promising local market structure for direct electricity exchange~\cite{tushar2020peer, kim2023pricing}. Although P2P trading is settled as economic contracts, its power flows are realized through distribution networks. Participant locations and trading volumes can therefore compromise voltage security, cause line congestion, and increase losses. Thus, network-aware coordination is required to align market outcomes with distribution system operating constraints~\cite{guerrero2019decentralized, paudel2020peer}.

Recent studies have proposed various coordination methods that integrate P2P energy trading with distribution system operation. Representative approaches use voltage sensitivity coefficients, power transfer distribution factors (PTDFs), or loss sensitivity factors to quantify the impact of each transaction on voltages, line capacities, and losses, and reflect these impacts in transaction limits, network usage charges, or market-clearing signals~\cite{guerrero2019decentralized, feng2023peer}. Other studies consider flexibility-market-based coordination, game-theoretic market clearing with network operators, or causality-based network cost allocation to encourage grid-aware peer matching~\cite{khorasany2022framework, belgioioso2022operationally, kim2024causality}. Although these approaches can incorporate network constraints into market mechanisms while preserving the decentralized nature of P2P trading, most of them require network sensitivity or equivalent network model information.

This reliance on network sensitivity creates a critical information-disclosure issue. PTDFs and voltage sensitivity are not merely operational auxiliary data; they can also serve as attack-optimization information because they indicate how an injection change at a specific bus affects line flows, voltages, and state-estimation results. Since more sensitive lines reach their limits faster under identical conditions, sensitivity information reveals vulnerable components~\cite{yan2014integrated}. Once exposed, an attacker may exploit historical data and partial network information to induce targeted line overloads, even without the full network model, as shown in false data injection (FDI) attack studies~\cite{zhang2018can}. Voltage-sensitivity and replay-attack studies show that this information can support attack-resource concentration, measurement selection, timing decisions, detection evasion, and maximized impact on state estimation~\cite{liu2025voltage, de2025replay}. Thus, disclosing sensitivity matrices, rankings, meter criticality, or bus/line impact factors can directly support target selection, attack-resource allocation, and stealthiness.

In the conventional distribution-utility model, the incumbent distribution utility owns and operates the local distribution network and often performs the distribution system operator (DSO), unless these functions are assigned to an independent DSO~\cite{demartini2015distribution, thomas2018evolution, warwick2016electricity}. When such a utility, acting in its DSO capacity, derives transaction guides, network charges, or transaction approval decisions from private network information, market participants cannot easily verify whether these outputs are honestly generated according to actual distribution system constraints. This issue is particularly important in P2P markets because prosumers can alter the conventional utility-centered electricity supply structure and may reduce energy sales revenues~\cite{kloppenburg2019digital, dippenaar2026synthesis}. Without verifiability, participants cannot distinguish legitimate security-oriented signals from strategic signals reflecting particular interests.

Blockchain technology has been proposed to reduce reliance on centralized market operators and to strengthen trust among prosumers. Existing studies use blockchain for immutable transactions records, smart-contract-based billing and settlement, bidding, market clearing, authentication, and consensus through mechanisms such as PBFT, IBFT, and PoA~\cite{esmat2021novel, abdella2021architecture, dong2022decentralized}. Privacy and security-enhancing techniques, including digital signatures, signcryption, permissioned validators, off-chain/side-chain structures, homomorphic encryption, and zero-knowledge proofs, have also been discussed, but mostly for transaction-record protection, user authentication, settlement-condition verification, and scalability~\cite{gao2023blockchain, veerasamy2024blockchain}. Hence, existing blockchain-based P2P energy trading studies have not fully exploited cryptographic verifiability to prove the honest execution of optimization or settlement computations without revealing sensitive inputs.

To bridge this gap, this paper proposes a zero-knowledge-proof-based method for verifying network-constrained transaction guidelines without exposing network sensitivity. A zero-knowledge proof can certify that a public transaction guide was computed correctly from committed private inputs according to an approved circuit~\cite{gmr1989, groth2016}, but it does not by itself certify that the committed private inputs are the true physical network data. Therefore, this paper focuses on the computational-integrity layer of transaction-guide verification. Comparisons among the key features considered in the existing literature against our proposed method are presented in Table~\ref{tab:literature_comparison}. The main contributions of this paper are listed as follows.
\begin{itemize}
  \item We formulate a transaction guide for P2P energy trading as a robust box of admissible active-power injection and withdrawal deviations. The guide enforces aggregate voltage limits and line-flow limits using sign-decomposed voltage and line-flow sensitivity. Unlike transaction-specific approval rules, the guide provides an ex-ante operating envelope that can be used by different distribution-level market-clearing mechanisms.
  \item We develop an arithmetic-circuit representation of the transaction-guide verification problem. The circuit binds committed network parameters, Y-bus construction, Jacobian construction, sensitivit computation, sign-split consistency, primal feasibility, and LP optimality into a single proof statement. This allows participants to verify the computational integrity of the published guide without observing private network sensitivities.
  \item We integrate the verification circuit with a Groth16-style zero-knowledge proof and a blockchain-based statement registry. The blockchain layer anchors the approved circuit, verification key, public-input hash, statement context, nonce, and verification result, thereby preventing public-input substitution, proof replay, and post-proof guide tampering.  
  \item We explicitly characterize the trust boundary of the proposed method. The proof certifies correct computation with respect to committed private network data, while the authenticity of the committed network model is treated as an external attestation requirement. This distinction clarifies the security guarantee provided by zero-knowledge verification in distribution-level market operations.
\end{itemize}

\begin{table}[t]
    \caption{Comparison of Existing Literature and the Proposed Method}
    \label{tab:literature_comparison}
    \centering
    \scriptsize
    \renewcommand{\arraystretch}{1.15}
    \setlength{\tabcolsep}{2pt}
    \begin{tabularx}{\linewidth}{
        >{\raggedright\arraybackslash}p{0.34\linewidth}
        >{\centering\arraybackslash}X
        >{\centering\arraybackslash}X
        >{\centering\arraybackslash}X
        >{\centering\arraybackslash}X}
    \toprule
    \makecell[c]{Refs.}
    & \makecell[c]{Network\\constraints}
    & \makecell[c]{Sensitivity}
    & \makecell[c]{Blockchain}
    & \makecell[c]{Verification} \\
    \midrule
    \cite{guerrero2019decentralized, paudel2020peer, feng2023peer, khorasany2022framework, belgioioso2022operationally, kim2024causality}
    & O & O & X & X \\
    \midrule
    \cite{esmat2021novel, abdella2021architecture, dong2022decentralized, gao2023blockchain, veerasamy2024blockchain}
    & X & X & O & X \\
    \midrule
    Proposed method
    & O & O & O & O \\
    \bottomrule
    \end{tabularx}
    \vspace{-0mm}
\end{table}

\section{Preliminaries}
Let $[n]\coloneqq\{1,\ldots,n\}$. The circuit field is denoted by $\mathbb F_r$, whereas the elliptic-curve base field is denoted by $\mathbb F_p$. Bold symbols denote vectors and matrices, and $x_i$ denotes the $i$-th component of $\bm x$.

An arithmetic circuit over $\mathbb F_r$ is represented by a rank-1 constraint system (R1CS). Let $\bm a=(a_0,a_1,\ldots,a_m)$ be the assignment vector with $a_0=1$, public inputs $a_1,\ldots,a_\ell$, and private witness or auxiliary variables in the remaining entries. For each primitive constraint $q$, R1CS coefficient vectors $\bm u_q,\bm v_q,\bm w_q$ satisfy
\begin{equation}
    \langle \bm u_q,\bm a\rangle\langle \bm v_q,\bm a\rangle
    = \langle \bm w_q,\bm a\rangle,\quad q=1,\ldots,n_c .
    \label{eq:R1CS_std}
\end{equation}
To obtain succinct verification, the R1CS is transformed into a quadratic arithmetic program (QAP). For distinct points $\omega_q\in\mathbb F_r$, define $t(X)=\prod_{q=1}^{n_c}(X-\omega_q)$ and interpolate polynomials $u_i(X),v_i(X),w_i(X)$ such that their values at $\omega_q$ equal the corresponding R1CS coefficients. With
\begin{equation}
    \label{eq:qap_uvw}
    \begin{alignedat}{2}
    U(X) &= \sum_i a_i u_i(X), 
    &\quad V(X) &= \sum_i a_i v_i(X),\\
    W(X) &= \sum_i a_i w_i(X).
    \end{alignedat}
\end{equation}
the assignment satisfies the circuit if and only if
\begin{equation}
    U(X)V(X)-W(X)=t(X)H(X)
\end{equation}
for some quotient polynomial $H(X)$~\cite{gennaro2013quadratic}.

The proposed circuit uses Poseidon as a circuit-efficient hash function~\cite{grassi2021poseidon}. A binary Merkle tree is built by repeatedly applying Poseidon to pairs of child hashes, and the resulting root commits to the network-parameter. Zero-knowledge verification is implemented with a Groth16-style pairing argument. Let $\mathbb G_1,\mathbb G_2$ be cyclic groups of prime order $r$, and let $e:\mathbb G_1\times\mathbb G_2\rightarrow\mathbb G_T$ be a non-degenerate bilinear pairing. Bilinearity, $e([a]P,[b]Q)=e(P,Q)^{ab}$, enables constant-size consistency checks in the common reference string~\cite{miller1986short,freeman2010taxonomy,groth2016size}. Elliptic-curve discrete-logarithm hardness provides the underlying computational assumption~\cite{koblitz1987elliptic,menezes1996handbook}.

\section{Transaction Guide}
We consider a power system represented by $\mathcal G=(\mathcal N,\mathcal L)$, where $\mathcal N:=\{1,\dots,N\}$ denotes the set of buses and $\mathcal L$ denotes the set of monitored directed lines.
Bus $N$ is selected as the slack bus, and the set of non-slack buses is $\mathcal N^\circ:=\mathcal N\setminus\{N\}=\{1,\dots,N-1\}$, with $n:=N-1$.
Each monitored line $l\in\mathcal L$ is associated with an ordered bus pair $(i,j)$.
For each bus $i\in\mathcal N$, $\mathcal N_i:=\{j\in\mathcal N\mid (i,j)\in\mathcal L \ \text{or} \ (j,i)\in\mathcal L\}$ denotes the set of buses adjacent to bus $i$.
Let $\mathcal S_i:=\mathcal N_i\cup\{i\}$.
The line admittance between buses $i$ and $j$ is denoted by $y_{ij}=g_{ij}+\mathrm{j}b_{ij}$.
The bus-admittance matrix is denoted by $\bm Y\in\mathbb C^{N\times N}$, with entries $Y_{ij}$.
For each bus $i$, $\theta_i$ and $|V_i|$ denote the voltage angle and magnitude, whereas $P_i$ and $Q_i$ denote the active and reactive power injections. 
For each monitored line $l=(i,j)\in\mathcal L$, $S_l^{\mathrm F}=P_l^{\mathrm F}+\mathrm{j}Q_l^{\mathrm F}$ denotes the sending-end complex power flow.

\subsection{Voltage Sensitivity}
Let the operating point be $\bm x^0:=[(\bm\theta^0)^\top,(\bm v^0)^\top]^\top \in\mathbb R^{2n}.$ 
The Jacobian matrix evaluated at $\bm x^0$ is
\begin{equation}
    \bm J^{0} :=
    \left. \frac{\partial(\bm P,\bm Q)}{\partial(\bm\theta,|\bm V|)}
    \right|_{\bm x=\bm x^0}
    \in\mathbb R^{2n\times 2n}.
    \end{equation}
Assuming that $\bm J^0$ is nonsingular, the local state sensitivity with respect to power-injection variations is
\begin{align}
    \begin{bmatrix}
    d\bm\theta\\
    d|\bm V|
    \end{bmatrix}
    =
    (\bm J^{0})^{-1}
    \begin{bmatrix}
    d\bm P\\
    d\bm Q
    \end{bmatrix}. 
    \label{eq:invJ}
\end{align}
Let the Jacobian be partitioned at $\bm x^0$ as
\begin{equation}
\bm J^0=
    \begin{bmatrix}
    \bm J^{11} & \bm J^{12}\\
    \bm J^{21} & \bm J^{22}
    \end{bmatrix}
    \in\mathbb R^{2n\times 2n}.
\end{equation}
Let $\bm E_P:=[\bm I_n,\bm 0]^\top$. We define the voltage-magnitude sensitivity with respect to active power as the lower $n$ rows of $(\bm J^0)^{-1}\bm E_P$:
\begin{equation}
    \bm A^{\mathrm{V}} :=
    \left.
    \frac{\partial |\bm V|}{\partial \bm P}
    \right|_{\bm x=\bm x^0}
    \in\mathbb R^{n\times n}.
    \label{eq:def_AV}
\end{equation}

\subsection{Line-Flow Sensitivity}
Define the apparent-power magnitude vector $|\bm S^{\mathrm F}|\in\mathbb R^L$ by
\begin{equation}
    |\bm S^{\mathrm F}|^{\odot 2} := \bm P^{\mathrm F}\odot \bm P^{\mathrm F} + \bm Q^{\mathrm F}\odot \bm Q^{\mathrm F},
    \label{eq:Sf_def}
\end{equation}
where $\odot$ denotes the Hadamard product. Differentiating~\eqref{eq:Sf_def}, the line-flow magnitude sensitivity with respect to active-power injections is defined as
\begin{equation}
\label{eq:def_AS}
\begin{aligned}
\bm A^{\mathrm F}
&:= 
\frac{\partial |\bm S^{\mathrm F}|}{\partial \bm P}
\Big|_{\bm x=\bm x^0} \\
&= \operatorname{diag}(|\bm S^{\mathrm F,0}|)^{-1}
\Big[
\operatorname{diag}(\bm P^{\mathrm F,0})
\frac{\partial \bm P^{\mathrm F}}{\partial \bm P}
\Big|_{\bm x=\bm x^0} \\
&\quad
+
\operatorname{diag}(\bm Q^{\mathrm F,0})
\frac{\partial \bm Q^{\mathrm F}}{\partial \bm P}
\Big|_{\bm x=\bm x^0}
\Big].
\end{aligned}
\end{equation}
where line-flow derivatives with respect to active-power injections are obtained from the chain rule and sensitivities in~\eqref{eq:invJ}.

\subsection{Robust Transaction Guide}
Assuming that reactive-power variations are negligible, the first-order voltage approximation satisfies
\begin{align}
    |\underline{\bm V}|
    \preceq
    |\bm V|^0+\bm A^{\mathrm{V}}\Delta\bm P
    \preceq
    |\overline{\bm V}|.
    \label{eq:V_linear}
\end{align}
Similarly, the line-flow magnitudes satisfy the upper-bound approximation
\begin{align}
    |\bm S^{\mathrm F}|^0+\bm A^{\mathrm{F}}\Delta\bm P
    \preceq
    \overline{|\bm S^{\mathrm F}|}.
    \label{eq:S_linear}
\end{align}

The transaction guide is defined as the box
\begin{equation}
    \mathcal B(\bm u,\bm\ell)
    :=
    \left\{
    \Delta\bm P\in\mathbb R^n\ \middle|\
    -\bm\ell\preceq\Delta\bm P\preceq\bm u
    \right\},
    \label{eq:buswise_box}
\end{equation}
where $\bm u,\bm\ell\in\mathbb R_+^n$ denote the upward and downward active-power transaction widths around the operating point. Let $\overline{\bm p},\underline{\bm p}\in\mathbb R_+^n$ denote upper
bounds on upward and downward transactions. To obtain a robust box that is valid for every $\Delta\bm P\in\mathcal B(\bm u,\bm\ell)$, we use the elementwise sign decomposition
\begin{align}
    \bm A^{\mathrm{V}}=\bm A^{\mathrm{V}+}-\bm A^{\mathrm{V}-},\qquad
    \bm A^{\mathrm{F}}=\bm A^{\mathrm{F}+}-\bm A^{\mathrm{F}-},
    \label{eq:AV_AS_signsplit}
\end{align}
where all sign-split matrices are elementwise nonnegative. The weighted maximum-width guide is obtained from
\begin{subequations}\label{eq:max_box_LP_vector}
\begin{align}
    \max_{\bm u,\bm\ell}\quad
    & \bm w^\top(\bm u+\bm\ell) \label{eq:max_box_obj}\\
    \text{s.t.}\quad
    & \bm 0\preceq \bm u\preceq \overline{\bm p},\qquad \bm 0\preceq \bm\ell\preceq \underline{\bm p}, \label{eq:max_box_cap_vec}\\
    & \bm 1^\top \bm u = \bm 1^\top \bm\ell, \label{eq:max_box_balance_vec}\\
    & |\bm V|^0-\bm A^{\mathrm{V}+}\bm\ell-\bm A^{\mathrm{V}-}\bm u
      \succeq |\underline{\bm V}|,
    \label{eq:max_box_Vlow_vec}\\
    & |\bm V|^0+\bm A^{\mathrm{V}+}\bm u+\bm A^{\mathrm{V}-}\bm\ell
      \preceq |\overline{\bm V}|,
    \label{eq:max_box_Vup_vec}\\
    & |\bm S^{\mathrm F}|^0+\bm A^{\mathrm{F}+}\bm u+\bm A^{\mathrm{F}-}\bm\ell
      \preceq \overline{|\bm S^{\mathrm F}|}.
    \label{eq:max_box_Sup_vec}
\end{align}
\end{subequations}
Problem~\eqref{eq:max_box_LP_vector} computes a transaction guide. The vector $\bm u$ denotes the maximum admissible injection at each participating bus, whereas $\bm\ell$ denotes the maximum admissible withdrawal. The balance constraint \eqref{eq:max_box_balance_vec} ensures that the total admissible injection and withdrawal volumes are equal. The constraints \eqref{eq:max_box_Vlow_vec}--\eqref{eq:max_box_Sup_vec} enforce robust voltage and line-flow security conditions under the sign-split sensitivity matrices.

\begin{lemma}[First-order security of the transaction guide]
Suppose that $\bm A^{\mathrm V}=\bm A^{\mathrm V+}-\bm A^{\mathrm V-}$ and $\bm A^{\mathrm F}=\bm A^{\mathrm F+}-\bm A^{\mathrm F-}$ are elementwise sign decompositions, where all sign-split matrices are elementwise nonnegative.
If $(\bm u,\bm\ell)$ satisfies Problem~\eqref{eq:max_box_LP_vector}, then every
$\Delta\bm p\in\mathcal B(\bm u,\bm\ell)$ satisfies~\eqref{eq:V_linear} and~\eqref{eq:S_linear}.
\end{lemma}
\begin{proof}
For any $\Delta\bm p\in\mathcal B(\bm u,\bm\ell)$, $-\bm\ell\preceq\Delta\bm p\preceq\bm u$. Since the sign-split matrices are elementwise nonnegative,
\begin{align}
    -\bm A^{\mathrm{V}+}\bm\ell-\bm A^{\mathrm{V}-}\bm u
    \preceq
    \bm A^{\mathrm{V}}\Delta\bm p
    \preceq
    \bm A^{\mathrm{V}+}\bm u+\bm A^{\mathrm{V}-}\bm\ell,
\end{align}
and
\begin{align}
    \bm A^{\mathrm{F}}\Delta\bm p
    \preceq
    \bm A^{\mathrm{F}+}\bm u+\bm A^{\mathrm{F}-}\bm\ell.
\end{align}
Combining these bounds with 
\eqref{eq:max_box_Vlow_vec}--\eqref{eq:max_box_Sup_vec} proves the claim.
\end{proof}
The transaction guide provides a robust operating envelope for upward and downward active-power deviations at each bus. The equality in~\eqref{eq:max_box_balance_vec} makes the published upward and downward guide volumes balance-aware, while the market-clearing layer selects the realized seller--buyer transaction within this envelope.

\section{Arithmetic Circuit Design}
This section constructs an arithmetic circuit for the proposed verification statement. All R1CS constraints are enforced over the circuit field $\mathbb F_r$.

\subsection{Verification Statement and Input Structure}
The public input vector includes the published guide $(\bm u^\star,\bm \ell^\star)$, the operating point, transaction caps, voltage and line-flow limits, and objective weights:
\begin{align}
\bm x_{\mathrm{pub}}:= \Big(R_Y, \bm\theta^0, |\bm V|^0, \overline{\bm p}, \underline{\bm p}, \bm u^\star, \bm\ell^\star, |\underline{\bm V}|, |\overline{\bm V}|, |\bm S^{\mathrm F}|^0, \overline{|\bm S^{\mathrm F}|}, \bm w \Big).
\label{eq:x_pub}
\end{align}
Thus, the proof is bound to the entire guide-computation instance, not only to the guide values. Any change in the guide, operating point, transaction caps, security limits, or objective weights alters the public-input accumulation term and invalidates a proof generated for another instance.

The private witness contains the network parameters and all intermediate variables required to verify sensitivity construction and guide optimality, which are not revealed to the verifier:
\begin{align}
\bm x_{\mathrm{priv}} := \Big(\bm G, \bm B, \bm Y, \bm J, \bm A, \bm A^{\mathrm{V}}, \bm A^{\mathrm{F}}, \bm A^{\mathrm{V}+},
\bm A^{\mathrm{V}-}, \bm A^{\mathrm{F}+}, \bm A^{\mathrm{F}-}, \bm\eta \Big),
\label{eq:x_wit}
\end{align}
where $(\bm\eta)$ collects the KKT multipliers, slack variables, and auxiliary variables introduced for intermediate products.

Given the public input vector in~\eqref{eq:x_pub}, acceptance of the proof implies that the prover has demonstrated that:
\begin{enumerate}
    \item there exist line parameters $\bm Y= \bm G + \mathrm{j}\bm B$ consistent with the commitment $R_Y$;
    \item the admittance matrix $\bm Y$ and the Jacobian matrix $\bm J$ are correctly constructed from the committed network and evaluated at the operating state $(\bm\theta^0,|\bm V|^0)$;
    \item the state-sensitivity matrix $\bm A$ satisfies the linear system
    $\bm J^0\bm A=\bm E$;
    \item the sensitivity matrices $(\bm A^{\mathrm{V}},\bm A^{\mathrm{F}})$, together with their
    sign-split forms $(\bm A^{\mathrm{V}+},\bm A^{\mathrm{V}-},\bm A^{\mathrm{F}+},\bm A^{\mathrm{F}-})$, are
    properly derived from $\bm A$; and
    \item $(\bm u^\star,\bm\ell^\star)$ is optimal for
    \eqref{eq:max_box_LP_vector}, as certified by valid KKT multipliers and
    slack variables.
\end{enumerate}

\subsection{Network-Parameter Commitment}
The committed network parameters are bound to the public root $R_Y$ through a Poseidon Merkle tree. For each monitored line $l=(i,j)\in\mathcal L$, the circuit forms the leaf
\begin{equation}
    h_l=\mathrm{Poseidon}(1,l,G_l,B_l),
    \label{eq:poseidon_network_leaf}
\end{equation}
where the leading value is a domain-separation tag. Internal Merkle nodes are computed with the same Poseidon compression rule as in Section~II, and the final root is constrained by
\begin{equation}
    R_Y^{\mathrm{calc}}-R_Y=0.
    \label{eq:network_root_binding}
\end{equation}
Thus, the line parameters used in the Y-bus, Jacobian, sensitivity, and guide-optimality constraints are the same parameters committed by the public input $R_Y$.

\subsection{Y-Bus Construction}
For each adjacent bus pair $(i,j)$ with $i,j\in\mathcal N$ and $j\in\mathcal N_i$, the bus-admittance construction is enforced by
\begin{subequations}
    \begin{align}
    Y_{ij}^G+G_l &= 0,  && l=(i,j)\in\mathcal L, \label{eq:r1cs_Y_sparse_Goff}\\
    Y_{ij}^B+B_l &= 0,  && l=(i,j)\in\mathcal L, \label{eq:r1cs_Y_sparse_Boff}\\
    Y_{ii}^G-\sum_{k\in\mathcal N_i}G_{ik} &= 0,  && i\in\mathcal N, \label{eq:r1cs_Y_sparse_Gdiag}\\
    Y_{ii}^B-\sum_{k\in\mathcal N_i}B_{ik} -2^{-1}B_{i,\mathrm{sh}} &= 0, && i\in\mathcal N. \label{eq:r1cs_Y_sparse_Bdiag}
    \end{align}
\end{subequations}

\subsection{Jacobian Construction}
Letting $\alpha_i:=|V_i^0|^2,\, \beta_{ik}:=|V_i^0||V_k^0|$ denote circuit constants, the bus power injections are defined by
\begin{subequations}
    \begin{align}
        \left(\alpha_iG_{ii}+\sum_{k\in\mathcal N_i}\beta_{ik}v_{ik}\right)\cdot 1 &= P_i,\\
        \left(-\alpha_iB_{ii}+\sum_{k\in\mathcal N_i}\beta_{ik}u_{ik}\right)\cdot 1 &= Q_i,
    \end{align}
\end{subequations}
where $v_{ik}=m^{Gc}_{ik}+m^{Bs}_{ik}$ and $u_{ik}=m^{Gs}_{ik}-m^{Bc}_{ik}$ are implemented by
\begin{subequations}
    \begin{align}
        G_{ik}\cdot c_{ik} &= m^{Gc}_{ik},\\
        B_{ik}\cdot s_{ik} &= m^{Bs}_{ik},\\
        G_{ik}\cdot s_{ik} &= m^{Gs}_{ik},\\
        B_{ik}\cdot c_{ik} &= m^{Bc}_{ik}.
    \end{align}
\end{subequations}
Here, $s_{ik}$ and $c_{ik}$ are trigonometric polynomial approximations defined in Appendix~\ref{App:trig}. The entries of the Jacobian $\bm J$ are represented by
\begin{equation}
    \label{eq:J_entries_cases}
    \begin{aligned}
    J_{ij}^{11} &=
        \begin{cases}
        \beta_{ij}u_{ij}, & i\neq j,\\
        -Q_i-\alpha_iB_{ii}, & i=j,
        \end{cases}
        \\
        J_{ij}^{12} &=
        \begin{cases}
        |V_i^0|v_{ij}, & i\neq j,\\
        P_i|V_i^0|^{-1}+G_{ii}|V_i^0|, & i=j,
        \end{cases}
        \\
        J_{ij}^{21} &=
        \begin{cases}
        -\beta_{ij}v_{ij}, & i\neq j,\\
        P_i-\alpha_iG_{ii}, & i=j,
        \end{cases}
        \\
        J_{ij}^{22} &=
        \begin{cases}
        |V_i^0|u_{ij}, & i\neq j,\\
        Q_i|V_i^0|^{-1}-B_{ii}|V_i^0|, & i=j.
        \end{cases}
    \end{aligned}
\end{equation}

\subsection{Network Sensitivity Binding}
Let $\mathcal N^\circ:=\mathcal N\setminus\{N\}$ denote the set of non-slack buses. We introduce the sensitivity matrix $\bm A$ as a witness and certify it by enforcing
\begin{equation}
    \bm J^0\bm A=\bm E,
    \label{eq:r1cs_JAP}
\end{equation}
where
\begin{equation}
    \bm A :=
    \begin{bmatrix}
        \bm A^\theta\\
        \bm A^{\mathrm V}
    \end{bmatrix}
    \in\mathbb F_r^{2n\times n},
    \qquad
    \bm E:=
    \begin{bmatrix}
        \bm I_n\\
        \bm 0
    \end{bmatrix}
    \in\mathbb F_r^{2n\times n}.
    \label{eq:r1cs_AP_EP}
\end{equation}
Here, $\bm A^\theta,\bm A^{\mathrm V}\in\mathbb F_r^{n\times n}$ encode the voltage-angle and voltage-magnitude sensitivities with respect to active-power injections.
For each $i,k\in\mathcal N^\circ$ and each $j\in\mathcal S_i$, the row-by-column product in \eqref{eq:r1cs_JAP} is enforced by introducing auxiliary product variables $m_{ijk}^{11}$, $m_{ijk}^{12}$, $m_{ijk}^{21}$, and $m_{ijk}^{22}$ as follows:
\begin{subequations}\label{eq:r1cs_AP_sparse1}
    \begin{align}
        J_{ij}^{11,0}\cdot A_{jk}^{\theta} &= m_{ijk}^{11},
        && \forall j\in\mathcal S_i,\\
        J_{ij}^{12,0}\cdot A_{jk}^{\mathrm V} &= m_{ijk}^{12},
        && \forall j\in\mathcal S_i,\\
        J_{ij}^{21,0}\cdot A_{jk}^{\theta} &= m_{ijk}^{21},
        && \forall j\in\mathcal S_i,\\
        J_{ij}^{22,0}\cdot A_{jk}^{\mathrm V} &= m_{ijk}^{22},
        && \forall j\in\mathcal S_i.
    \end{align}
\end{subequations}
The corresponding row-sum constraints are imposed for all $i,k\in\mathcal N^\circ$:
\begin{subequations}\label{eq:r1cs_AP_sparse2}
    \begin{align}
        \sum_{j\in\mathcal S_i}
        \left(m_{ijk}^{11}+m_{ijk}^{12}\right) &= \delta_{ik}, \label{eq:r1cs_AP_sparse2_P}\\
        \sum_{j\in\mathcal S_i} \left(m_{ijk}^{21}+m_{ijk}^{22}\right) &= 0,
        \label{eq:r1cs_AP_sparse2_Q}
    \end{align}
\end{subequations}
Here, $\delta_{ik}$ denotes the Kronecker delta, i.e., $\delta_{ik}=1$ if $i=k$ and $\delta_{ik}=0$ otherwise.

The line-flow magnitude sensitivity matrix $\bm A^{\mathrm F}$ is then certified from the sensitivity matrix $\bm A$ through the chain rule. 
For each $l\in\mathcal L$ and $k\in\mathcal N^\circ$, let
\begin{subequations}\label{eq:def_Sigma_flow_sens}
    \begin{align}
        \Sigma_{lk}^{P} := \left.\frac{\partial P_l^{\mathrm F}}{\partial P_k}\right|_{\bm x=\bm x^0},\quad
        \Sigma_{lk}^{Q} := \left.\frac{\partial Q_l^{\mathrm F}}{\partial P_k}\right|_{\bm x=\bm x^0}.
    \end{align}
\end{subequations}
Assuming that $|S_l^{\mathrm F,0}|\neq 0$ for all monitored lines $l\in\mathcal L$, the apparent-power magnitude sensitivity
\begin{align}
    A_{lk}^{\mathrm F} :=
    \left.
    \frac{\partial |S_l^{\mathrm F}|}{\partial P_k}
    \right|_{\bm x=\bm x^0}
\end{align}
satisfies the equivalent chain-rule relation
\begin{equation}
    A_{lk}^{\mathrm F}|S_l^{\mathrm F,0}|
    =
    P_l^{\mathrm F,0}\Sigma_{lk}^{P}
    +
    Q_l^{\mathrm F,0}\Sigma_{lk}^{Q}.
    \label{eq:flow_mag_chain_rule}
\end{equation}
This relation is enforced in R1CS form by introducing auxiliary variables
$n_{lk}$, $n_{lk}^{P}$, and $n_{lk}^{Q}$:
\begin{subequations}\label{eq:r1cs_AF_final}
    \begin{align}
        A_{lk}^{\mathrm F}\cdot |S_l^{\mathrm F,0}| &= n_{lk}, && \forall l\in\mathcal L,\ \forall k\in\mathcal N^\circ,\\
        P_l^{\mathrm F,0}\cdot \Sigma_{lk}^{P} &= n_{lk}^{P}, && \forall l\in\mathcal L,\ \forall k\in\mathcal N^\circ,\\
        Q_l^{\mathrm F,0}\cdot \Sigma_{lk}^{Q} &= n_{lk}^{Q}, && \forall l\in\mathcal L,\ \forall k\in\mathcal N^\circ,\\
        n_{lk}^{P}+n_{lk}^{Q}-n_{lk} &= 0, && \forall l\in\mathcal L,\ \forall k\in\mathcal N^\circ.
    \end{align}
\end{subequations}
The line-flow equations and their sensitivity expressions used in the R1CS constraints are provided in Appendices~\ref{App:flowEquation} and~\ref{App:lineflowsens}.

\subsection{Sign-Split Consistency Constraints}
We enforce the sign-split identities and elementwise complementarity:
\begin{subequations}\label{eq:r1cs_signsplit}
    \begin{align}
    A_{ik}^{V}-A_{ik}^{V+}+A_{ik}^{V-} &= 0, && \forall i,k\in \mathcal N^\circ,\\
    A_{lk}^{F}-A_{lk}^{F+}+A_{lk}^{F-} &= 0, && \forall l\in \mathcal L,\ \forall k\in \mathcal N^\circ,\\
    A_{ik}^{V+}\cdot A_{ik}^{V-} &= 0, && \forall i,k\in \mathcal N^\circ,\\
    A_{lk}^{F+}\cdot A_{lk}^{F-} &= 0, && \forall l\in \mathcal L,\ \forall k\in \mathcal N^\circ.
    \end{align}
\end{subequations}
To enforce nonnegativity of the sign-split components, we apply range checks. For a value $x$, let $\mathrm{RC}_\kappa(x)$ denote the $\kappa$-bit range check implemented using auxiliary bits $\{b_q\}_{q=0}^{\kappa-1}$:
\begin{subequations}\label{eq:range_check}
    \begin{align}
    b_q(b_q-1) &= 0,
    && \forall q\in\{0,\dots,\kappa-1\},\\
    x-\sum_{q=0}^{\kappa-1}2^q b_q &= 0.
    \end{align}
\end{subequations}
We impose $\mathrm{RC}_\kappa(A_{ik}^{V+}),\; \mathrm{RC}_\kappa(A_{ik}^{V-}),\; \mathrm{RC}_\kappa(A_{lk}^{F+}),\; \mathrm{RC}_\kappa(A_{lk}^{F-})$ for all admissible indices.

\subsection{Transaction-Guide Optimality}
Let $\underline{\rho}^u_i$ and $\overline{\rho}^u_i$ be the dual variables for $u_i\ge 0$ and $u_i\le \overline{p}_i$, and let $\underline{\rho}^{\ell}_i$ and $\overline{\rho}^{\ell}_i$ be the dual variables for $\ell_i\ge 0$ and $\ell_i\le \overline{p}^{\ell}_i$, The voltage-bound dual variables are denoted by $\underline{\mu}_k$ and $\overline{\mu}_k$, the line-flow dual variable by $\nu_l$, and the guide-balance multiplier by $\lambda_{\mathrm b}$. The published guide $(\bm u^\star,\bm\ell^\star)$ is the optimal solution of the maximum-box LP in~\eqref{eq:max_box_LP_vector}, whose optimality is certified by the Karush--Kuhn--Tucker conditions

\subsubsection{Primal Feasibility}
To encode the matrix--vector products, we introduce auxiliary variables 
\[
\Pi_{k,i}^{V+,\ell},\ \Pi_{k,i}^{V-,u},\ \Pi_{k,i}^{V+,u},\ \Pi_{k,i}^{V-,\ell},\ \Pi_{l,i}^{F+,u},\ \Pi_{l,i}^{F-,\ell}\in\mathbb F_r.
\]
They satisfy
\begin{subequations}\label{eq:r1cs_mvm_products}
    \begin{align}
    A_{k,i}^{V+}\ell_i^\star &= \Pi_{k,i}^{V+,\ell}, && \forall k,i\in\mathcal N^\circ,\\
    A_{k,i}^{V-}u_i^\star &= \Pi_{k,i}^{V-,u},  && \forall k,i\in\mathcal N^\circ,\\
    A_{k,i}^{V+}u_i^\star &= \Pi_{k,i}^{V+,u},  && \forall k,i\in\mathcal N^\circ,\\
    A_{k,i}^{V-}\ell_i^\star &= \Pi_{k,i}^{V-,\ell},  && \forall k,i\in\mathcal N^\circ,\\
    A_{l,i}^{F+}u_i^\star &= \Pi_{l,i}^{F+,u},  && \forall l\in\mathcal L,\ \forall i\in\mathcal N^\circ,\\
    A_{l,i}^{F-}\ell_i^\star &= \Pi_{l,i}^{F-,\ell},  && \forall l\in\mathcal L,\ \forall i\in\mathcal N^\circ.
    \end{align}
\end{subequations}
Their row sums are
\begin{subequations}\label{eq:r1cs_mvm_sums}
    \begin{align}
    \sum_{i\in\mathcal N^\circ}\Pi_{k,i}^{V+,\ell} +\sum_{i\in\mathcal N^\circ}\Pi_{k,i}^{V-,u} &= t_k^{\min}, && \forall k\in\mathcal N^\circ,\\
    \sum_{i\in\mathcal N^\circ}\Pi_{k,i}^{V+,u} +\sum_{i\in\mathcal N^\circ}\Pi_{k,i}^{V-,\ell} &= t_k^{\max}, && \forall k\in\mathcal N^\circ,\\
    \sum_{i\in\mathcal N^\circ}\Pi_{l,i}^{F+,u} +\sum_{i\in\mathcal N^\circ}\Pi_{l,i}^{F-,\ell} &= \zeta_l^{\max}, && \forall l\in\mathcal L.
    \end{align}
\end{subequations}
The guide-balance equality is enforced by
\begin{equation}
    \sum_{i\in\mathcal N^\circ}u_i^\star-\sum_{i\in\mathcal N^\circ}\ell_i^\star=0.
    \label{eq:r1cs_guide_balance}
\end{equation}

Introducing nonnegative slack variables
$\bm s_u,\bm s_\ell,\bm s_V^{\min},\bm s_V^{\max},\bm s_S^{\max}\succeq 0$,
we enforce
\begin{subequations}\label{eq:r1cs_K1_full}
    \begin{align}
        u_i^\star+s_{u,i}-\overline p_i &= 0, && \forall i\in\mathcal N^\circ,\\
        \ell_i^\star+s_{\ell,i}-\underline p_i &= 0, && \forall i\in\mathcal N^\circ,\\
        s_{V,k}^{\min}+|\underline V|_k-|V|_k^0+t_k^{\min} &= 0, && \forall k\in\mathcal N^\circ,\\
        s_{V,k}^{\max}-|\overline V|_k+|V|_k^0+t_k^{\max} &= 0, && \forall k\in\mathcal N^\circ,\\
        s_{S,l}^{\max}-\overline{|S_l^{\mathrm F}|}+|S_l^{\mathrm F,0}|+\zeta_l^{\max} &= 0, && \forall l\in\mathcal L.
    \end{align}
\end{subequations}

The nonnegativity of the guide variables and slack variables is enforced as
\begin{subequations}\label{eq:K1_ranges_full}
    \begin{align}
        &\mathrm{RC}_\kappa(u_i^\star), \mathrm{RC}_\kappa(\ell_i^\star), \mathrm{RC}_\kappa(s_{u,i}), \mathrm{RC}_\kappa(s_{\ell,i}), && \forall i\in\mathcal N^\circ\!\!\!, \label{eq:K1_ranges_guide}\\
        &\mathrm{RC}_\kappa(s_{V,i}^{\min}), \mathrm{RC}_\kappa(s_{V,i}^{\max}), && \forall i\in\mathcal N^\circ\!\!, \label{eq:K1_ranges_voltage}\\
        &\mathrm{RC}_\kappa(s_{S,l}^{\max}), && \forall l\in\mathcal L,
        \label{eq:K1_ranges_line}
    \end{align}
\end{subequations}
Participant buses are modeled by nonzero transaction caps; nonparticipant buses have zero caps.

\subsubsection{Dual Feasibility}

The nonnegativity of the dual variables associated with inequality constraints is certified using the same $\kappa$-bit range constraints. 
The equality multiplier $\lambda_{\mathrm b}^\star$ is unrestricted.
Specifically, the guide-cap dual variables are constrained as
\begin{subequations}\label{eq:K2_ranges_full}
    \begin{align}
        &RC_\kappa(\underline{\rho}^{u}_i), RC_\kappa(\overline{\rho}^{u}_i), RC_\kappa(\underline{\rho}^{\ell}_i), RC_\kappa(\overline{\rho}^{\ell}_i),\quad \forall i\in\mathcal{N}^{\circ}, \\
        &RC_\kappa(\underline{\mu}_k), RC_\kappa(\overline{\mu}_k),\quad \forall k\in\mathcal{N}^{\circ},\\
        &RC_\kappa(\nu_l), \forall l\in\mathcal{L}.
    \end{align}
\end{subequations}

\subsubsection{Stationarity}
To express the stationarity conditions in R1CS form, we introduce auxiliary product variables $\Gamma_{k,i}^{V-,\underline\mu},\ \Gamma_{k,i}^{V+,\underline\mu},\ \Gamma_{k,i}^{V+,\overline\mu},\ \Gamma_{k,i}^{V-,\overline\mu},\ \Gamma_{l,i}^{F+,\overline\nu},\ \Gamma_{l,i}^{F-,\overline\nu}\in\mathbb F_r .$
These variables are constrained as
\begin{subequations}\label{eq:Gamma_defs}
\begin{align}
    A_{k,i}^{V-}\cdot\underline{\mu}_k^\star &= \Gamma_{k,i}^{V-,\underline\mu}, && \forall k,i\in\mathcal N^\circ,\\
    A_{k,i}^{V+}\cdot\underline{\mu}_k^\star &= \Gamma_{k,i}^{V+,\underline\mu}, && \forall k,i\in\mathcal N^\circ,\\
    A_{k,i}^{V+}\cdot\overline{\mu}_k^\star  &= \Gamma_{k,i}^{V+,\overline\mu},  && \forall k,i\in\mathcal N^\circ,\\
    A_{k,i}^{V-}\cdot\overline{\mu}_k^\star  &= \Gamma_{k,i}^{V-,\overline\mu},  && \forall k,i\in\mathcal N^\circ,\\
    A_{l,i}^{F+}\cdot\overline{\nu}_l^\star  &= \Gamma_{l,i}^{F+,\overline\nu},  && \forall l\in\mathcal L,\ \forall i\in\mathcal N^\circ,\\
    A_{l,i}^{F-}\cdot\overline{\nu}_l^\star  &= \Gamma_{l,i}^{F-,\overline\nu},  && \forall l\in\mathcal L,\ \forall i\in\mathcal N^\circ.
\end{align}
\end{subequations}

The stationarity conditions with respect to $u_i$ and $\ell_i$ are then enforced for all $i\in\mathcal N^\circ$ as
\begin{subequations}\label{eq:r1cs_K3_full}
    \begin{align}
        -w_i -\underline{\rho}_{u,i}^\star &+\overline{\rho}_{u,i}^\star -\lambda_{\mathrm b}^\star \notag \\ 
        &+\sum_{v\in\mathcal N^\circ}\Gamma_{v,i}^{V-,\underline\mu} -\sum_{v\in\mathcal N^\circ}\Gamma_{v,i}^{V+,\overline\mu} +\sum_{l\in\mathcal L}\Gamma_{l,i}^{F+,\overline\nu} =0, \label{eq:stationarity_u}\\
        -w_i -\underline{\rho}_{\ell,i}^\star &+\overline{\rho}_{\ell,i}^\star +\lambda_{\mathrm b}^\star \notag \\ 
        &+\sum_{v\in\mathcal N^\circ}\Gamma_{v,i}^{V+,\underline\mu} +\sum_{v\in\mathcal N^\circ}\Gamma_{v,i}^{V-,\overline\mu} +\sum_{l\in\mathcal L}\Gamma_{l,i}^{F-,\overline\nu} =0.
    \label{eq:stationarity_l}
    \end{align}
\end{subequations}

\subsubsection{Complementary Slackness}
The complementary-slackness conditions are enforced for the inequality constraints as follows:
\begin{subequations}\label{eq:r1cs_K4_full}
    \begin{align}
        &\underline{\rho}^{u}_i u_i^\star=0,\quad \overline{\rho}^{u}_i s_{u,i}=0,\quad \underline{\rho}^{\ell}_i \ell_i^\star=0,\quad \overline{\rho}^{\ell}_i s_{\ell,i}=0,\quad \forall i,\\
        &\underline{\mu}_k s^{\min}_{V,k}=0,\quad \overline{\mu}_k s^{\max}_{V,k}=0,\quad \forall k, \\
        &\nu_l s^{\max}_{S,l}=0,\quad \forall l.        
        \label{eq:cs_S_upper}
    \end{align}
\end{subequations}

\section{QAP Realization of the Verification Circuit}
All variables in the proposed circuit are scalar coordinates of
\begin{equation}
    \bm a=(1,\bm x_{\mathrm{pub}},\bm x_{\mathrm{priv}},\bm x_{\mathrm{aux}})\in\mathbb F_r^{M+1},
    \label{eq:qap_assignment}
\end{equation}
where $\bm x_{\mathrm{aux}}$ collects Poseidon states, Merkle nodes, trigonometric-approximation variables, products, range-check bits, and KKT auxiliary variables. Matrix quantities, including $\bm Y$, $\bm J$, $\bm A$, $\bm A^{\mathrm V}$, and $\bm A^{\mathrm F}$, are vectorized entrywise. The complete primitive-constraint set is partitioned as
\begin{align}
\begin{split}
    \mathcal Q =&\ \mathcal Q_{\mathrm{Poseidon}}\cup\mathcal Q_{\mathrm{Merkle}}\cup\mathcal Q_Y\cup\mathcal Q_{\mathrm{Trig}}\cup\mathcal Q_{\mathrm{PF}}\\
    &\cup\mathcal Q_J\cup\mathcal Q_A\cup\mathcal Q_{A_F}\cup\mathcal Q_{\mathrm{Split}}\cup\mathcal Q_{\mathrm{Range}}\cup\mathcal Q_{\mathrm{KKT}} .
\end{split}
\label{eq:qap_constraint_partition}
\end{align}
Each element of $\mathcal Q$ is expanded into the R1CS form in~\eqref{eq:R1CS_std}. The corresponding QAP is constructed by interpolation over $n_c=|\mathcal Q|$ points, and circuit satisfaction is equivalent to the divisibility condition
\begin{equation}
    U(X)V(X)-W(X)=t(X)H(X),
    \label{eq:qap_divisibility}
\end{equation}
where $U,V,W,t$, and $H$ follow the definitions in Section~II. Therefore, a valid Groth16 proof certifies, in one statement, network-parameter commitment, sensitivity binding, sign-split consistency, guide feasibility, and guide optimality.

\section{Zero-Knowledge Proof Procedure}
This section describes the zero-knowledge proof procedure. As shown in Fig.~\ref{fig:flowDiagram}, the off-chain layer performs circuit construction, circuit review, MPC-based setup, guide computation, input construction, and Groth16 proof generation. The blockchain layer stores the approved circuit record, anchors transaction-specific statement commitments, verifies submitted proofs, prevents replay, and executes the transaction only after successful verification. This design keeps confidential network data off-chain while giving participants a common, tamper-evident verification interface.
\begin{figure}[ht]
    \centering
    \includegraphics[width=1.0\columnwidth]{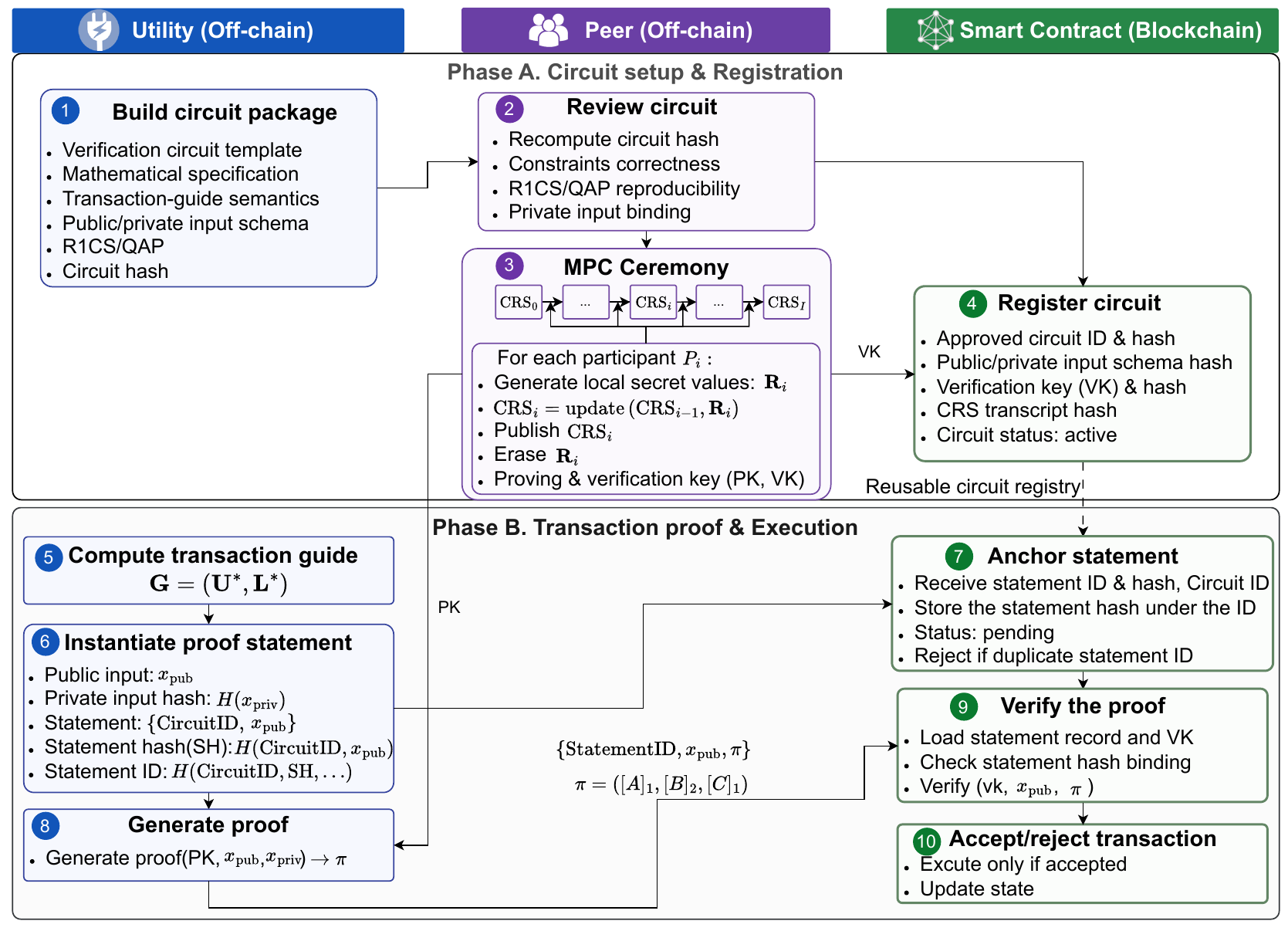}
    \vspace{-6mm}
    \caption{Zero-knowledge proof workflow for the proposed proof. Phase~A registers a reusable circuit and its
    verification key, whereas Phase~B anchors a transaction guide statement, generates a Groth16 proof, and executes the transaction only after successful verification.}
    \label{fig:flowDiagram}
\end{figure}
\subsection{Circuit Setup and Registration}
The utility first constructs a circuit package containing the circuit template, mathematical specification, bus-wise transaction-guide semantics, input schema, and R1CS/QAP representation. The circuit encodes the feasibility and optimality conditions of
\begin{align}
    \mathcal B(\bm u,\bm\ell) = \left\{\Delta\bm P\in\mathbb R^n\ \middle|\
    -\bm\ell\preceq \Delta\bm P\preceq \bm u \right\},
\end{align}
and verifies that the published guide $(\bm u^\star,\bm\ell^\star)$ is feasible and optimal for~\eqref{eq:max_box_LP_vector}. The circuit package is hashed as
\begin{align}
    h_C
    =
    H_{\mathsf{CIR}}
    \left(
    \mathsf{circuit\_package}
    \right).
\end{align}
The peers review the circuit by recomputing $h_C$, checking constraint correctness, verifying R1CS/QAP reproducibility, and confirming the intended guide statement. After approval, they execute an MPC ceremony to generate a circuit-specific common reference string (CRS),
\begin{align}
    \mathrm{CRS}_{\mathsf{cid}} = \left(\mathrm{pk}_{\mathsf{cid}},\mathrm{vk}_{\mathsf{cid}}\right),
\end{align}
where $\mathrm{pk}_{\mathsf{cid}}$ is the proving key used off-chain by the utility and $\mathrm{vk}_{\mathsf{cid}}$ is the verification key used by the smart contract. The smart contract stores a circuit record
\begin{align}
    \mathsf{CircuitRecord}[\mathsf{cid}] =\left(h_C,\; h_{\mathsf{io}},\; \mathrm{vk}_{\mathsf{cid}},\; h_{\mathsf{vk}},\; h_{\mathsf{CRS}},\; \mathsf{active} \right),
\end{align}
where $h_{\mathsf{io}}$ is the input-schema hash, $h_{\mathsf{vk}}$ is the verification-key hash, and $h_{\mathsf{CRS}}$ is the CRS transcript hash. This step stores only reusable circuit-level information.

\subsection{Statement Anchoring and Proof Generation}
For each transaction instance, the utility computes the published bus-wise guide $(\bm u^\star,\bm\ell^\star)$ and forms the public input vector $\bm x_{\mathrm{pub}}^{\mathrm{TG}}$ according to~\eqref{eq:x_pub}. To bind this large public input vector to an on-chain statement, the utility first computes
\begin{align}
    h_{\mathsf{pub}} = H_{\mathsf{PUB}} \left(\bm x_{\mathrm{pub}}^{\mathrm{TG}} \right),
    \label{eq:public_input_hash}
\end{align}
where $H_{\mathsf{PUB}}$ is a domain-separated hash function for the public input vector. Let $h_{\mathsf{ctx}}$ denote a transaction-context hash, $n_{\mathsf{tx}}$ a nonce, and $d$ a deadline. The transaction statement is defined as 
\begin{align}
    \mathsf{stmt} = \left(\mathsf{cid},\; h_{\mathsf{pub}},\; h_{\mathsf{ctx}},\; n_{\mathsf{tx}},\; d\right).
\end{align}
The statement commitment is then computed as
\begin{align}
    h_{\mathsf{stmt}}
    =
    H_{\mathsf{STMT}}
    \left(
    \mathsf{cid},\;
    h_{\mathsf{pub}},\;
    h_{\mathsf{ctx}},\;
    n_{\mathsf{tx}},\;
    d
    \right).
    \label{eq:statement_hash}
\end{align}
The lookup identifier is derived as
\begin{align}
    \mathsf{sid}
    =
    H_{\mathsf{SID}}
    \left(
    \mathsf{owner},\mathsf{cid},h_{\mathsf{stmt}},s
    \right),
    \label{eq:statement_id}
\end{align}
where $\mathsf{owner}$ is the transaction owner and $s$ is a salt. The
identifier $\mathsf{sid}$ is only a lookup key, whereas $h_{\mathsf{stmt}}$ binds the circuit identifier, public input vector, and transaction context.

The statement is anchored by calling
\begin{align}
    \mathsf{AnchorStatement}
    \left(
    \mathsf{sid},\mathsf{cid},h_{\mathsf{stmt}},d,n_{\mathsf{tx}}
    \right).
\end{align}
If $\mathsf{sid}$ is not duplicated and the circuit is active, the smart contract
stores
\begin{align}
\mathsf{StatementRecord}[\mathsf{sid}]
=
\left(
\mathsf{cid},\;
h_{\mathsf{stmt}},\;
\mathsf{owner},\;
d,\;
n_{\mathsf{tx}},\;
\mathsf{pending},\;
\mathsf{unused}
\right).
\end{align}
This record ensures that a later proof corresponds to the same circuit, public input vector, and transaction context committed in~\eqref{eq:statement_hash}.

The utility then generates a Groth16 proof
\begin{align}
    \pi
    \leftarrow
    \mathsf{Prove}
    \left(
    \mathrm{pk}_{\mathsf{cid}},
    \bm x_{\mathrm{pub}}^{\mathrm{TG}},
    \bm x_{\mathrm{wit}}^{\mathrm{TG}}
    \right),
    \label{eq:zk_prove}
\end{align}
where the resulting proof is
\begin{align}
    \pi
    =
    \left(
    [A]_1,[B]_2,[C]_1
    \right).
\end{align}
The private input $\bm x_{\mathrm{wit}}^{\mathrm{TG}}$ is never sent to the smart contract.

\subsection{On-chain Verification and Execution}
The utility submits a proof-carrying transaction
\begin{align}
    \mathsf{ExecuteWithProof}
    \left(
    \mathsf{sid},\bm x_{\mathrm{pub}}^{\mathrm{TG}},
    h_{\mathsf{ctx}},n_{\mathsf{tx}},d,\pi
    \right).
\end{align}
The smart contract loads the statement and circuit records, checks ownership or authorization, expiration, pending status, and nonce use, and then recomputes the public input hash and statement hash:
\begin{align}
    h_{\mathsf{pub}}'
    &=
    H_{\mathsf{PUB}}
    \left(
    \bm x_{\mathrm{pub}}^{\mathrm{TG}}
    \right),\\
    h_{\mathsf{stmt}}'
    &=
    H_{\mathsf{STMT}}
    \left(
    \mathsf{cid},\;
    h_{\mathsf{pub}}',\;
    h_{\mathsf{ctx}},\;
    n_{\mathsf{tx}},\;
    d
    \right).
\end{align}
The submitted public input is accepted only if
\begin{align}
    h_{\mathsf{stmt}}'
    =
    \mathsf{StatementRecord}[\mathsf{sid}].h_{\mathsf{stmt}}.
    \label{eq:statement_binding_check}
\end{align}
This statement-binding check prevents the prover from anchoring one guide statement and later submitting a proof for a different public input vector or context. After this check, the smart contract verifies the Groth16 proof:
\begin{align}
    \mathsf{Verify}
    \left(
    \mathrm{vk}_{\mathsf{cid}},
    \bm x_{\mathrm{pub}}^{\mathrm{TG}},
    \pi
    \right)
    \in
    \{0,1\}.
\end{align}
Internally, the verifier computes the public-input accumulation term
\begin{align}
    [X]_1
    =
    [\Psi_0]_1
    +
    \sum_{i=1}^{n_{\mathsf{pub}}}
    a_i[\Psi_i]_1,
\end{align}
where $a_1,\ldots,a_{n_{\mathsf{pub}}}$ are the public input elements. It then checks
\begin{align}
    \mathsf{e}([A]_1,[B]_2)
    \stackrel{?}{=}
    \mathsf{e}([\alpha]_1,[\beta]_2)
    \cdot
    \mathsf{e}([X]_1,[\gamma]_2)
    \cdot
    \mathsf{e}([C]_1,[\delta]_2),
    \label{eq:verifeq}
\end{align}
where $\mathsf{e}: \mathbb G_1\times\mathbb G_2\rightarrow\mathbb G_T$ is the bilinear pairing. If both~\eqref{eq:statement_binding_check} and \eqref{eq:verifeq} hold, the smart contract marks $\mathsf{sid}$ and $n_{\mathsf{tx}}$ as used and executes the transaction. Otherwise, the transaction is rejected and no state transition is performed.

\section{Case Study}
\subsection{Simulation Setup}
The proposed method is evaluated on a modified IEEE 33-node radial distribution system~\cite{dolatabadi2020enhanced}, with seller and buyer peers shown in Fig.~\ref{fig:caseNetwork}. The operating point is obtained from an AC power-flow solution under the base-load condition. At this point, $\bm A^{\mathrm V}$ and $\bm A^{\mathrm F}$ are computed and used to construct the robust transaction guide. The distributed market-clearing process is modeled using the dual-gradient method in~\cite{kim2024causality}.

\begin{figure}[ht]
    \vspace{0mm}
    \centering    
    \includegraphics[width=1.0\columnwidth]{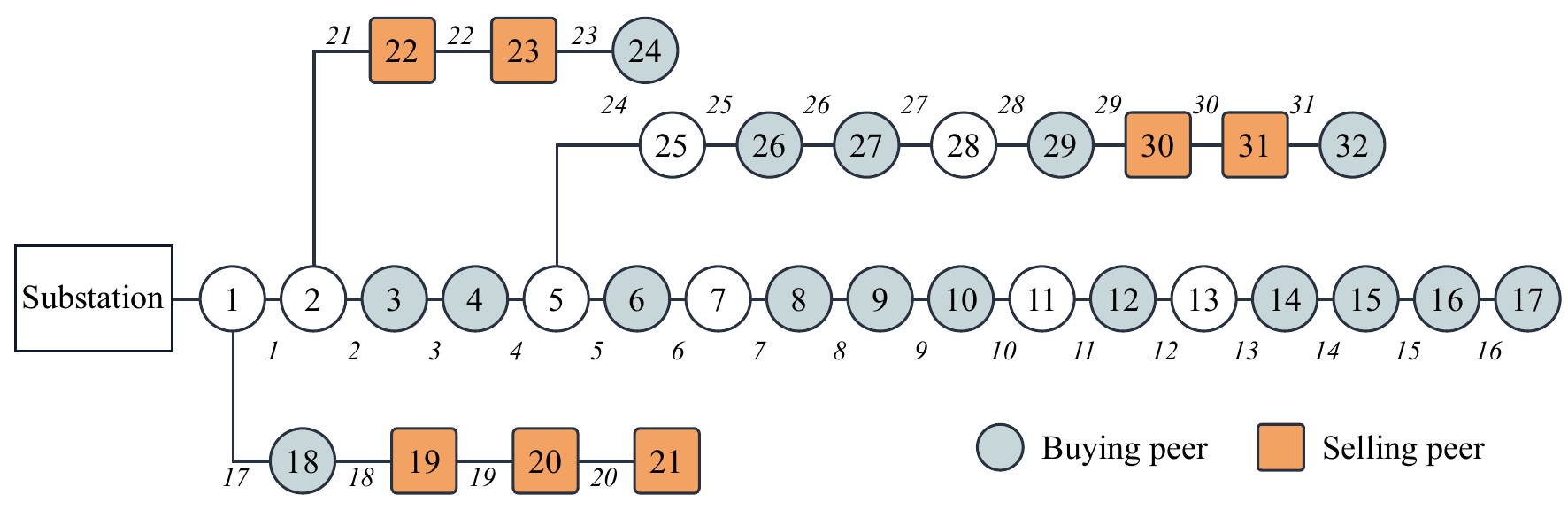}
    \vspace{-6mm}
    \caption{A modified IEEE 33-node distribution network with nodal indices in circles and squares. Italic numbers represent line indices.\vspace{0mm}}
    \label{fig:caseNetwork}
    \vspace{0mm}
\end{figure}
\begin{figure}[ht]
    \centering
     \subfigure[Nodal voltage magnitudes]{%
        \includegraphics[width=1.0\linewidth]{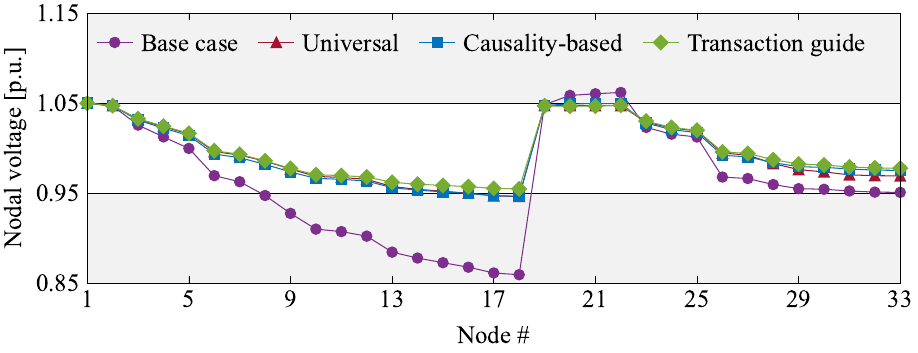}
        \label{fig:voltage}}\vfill
        \vspace{-2mm}
     \subfigure[Line loading (\% of flow limit)]{%
        \includegraphics[width=1.0\linewidth]{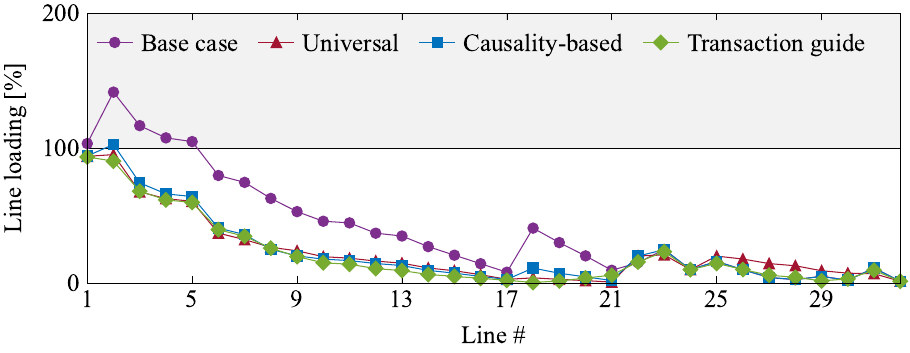}
        \label{fig:lineloading}}\vfill
        \vspace{-2mm}
    \caption{Nodal voltage and line-loading profiles under the base case, universal policy, causality-based policy, and transaction-guide policy.}
    \label{fig:volandline}
    \vspace{-0mm}   
\end{figure}

Four market-operation cases are compared:
\begin{enumerate}
    \item \textbf{Base case}: market clearing without network-security correction or transaction-guide limits;
    \item \textbf{Universal policy}: a uniform penalty is applied when network violations are observed;
    \item \textbf{Causality-based policy}: voltage and line-flow sensitivity terms are used to update participant-specific network charges;
    \item \textbf{Transaction-guide policy}: the proposed bus-wise transaction guide is imposed as an ex-ante operating envelope on the market-clearing process.
\end{enumerate}

The case study also evaluates utility-side tampering, where the utility modifies the guide after proof generation while reusing the original Groth16 proof. The verifier checks whether the modified public statement remains consistent with the certified computation. For verification, the computed guide is converted into fixed-point integers using a scaling factor of $10^6$. The R1CS instance, input, and proof artifacts are generated using the PySNARK/snarkJS workflow over the BN254 scalar field. The Groth16 proof is verified off-chain using \texttt{snarkjs} and on-chain in a local Hardhat EVM environment with a Solidity verifier and proof-registry contract.
\begin{table*}[t]
\centering
\caption{Utility-side transaction-guide tampering scenarios, modified public inputs, and market impacts}
\label{tab:utility_tamper_scenarios}
\footnotesize
\setlength{\tabcolsep}{3pt}
\renewcommand{\arraystretch}{1.15}
\begin{tabular*}{\textwidth}{
    @{\extracolsep{\fill}}
    p{0.14\textwidth}
    p{0.29\textwidth}
    p{0.28\textwidth}
    ccc@{}}
\hline
\multirow{2}{0.14\textwidth}{Scenario}
& \multirow{2}{0.29\textwidth}{Tampering action}
& \multirow{2}{0.28\textwidth}{Tampered public input}
& Trade 
& $\Delta$Trade 
& $\Delta P_{\rm loss}$ \\
& & & [MW] & [MW] & [MW] \\
\hline

Honest TG & No modification is applied to the certified transaction guide. & None & 1.570 & 0.000 & 0.000 \\

Low-cost seller restriction & The seller-side guide at bus 31 is reduced from $0.650$ MW to $0.000$ MW. & Seller guide at bus 31, public-input entry 22: $650000 \rightarrow 0$. & 1.184 & $-0.386$ & 0.006 \\

Buyer guide restriction & The buyer-side guide at bus 9 is reduced from $0.300$ MW to $0.000$ MW. &
Buyer guide at bus 9, public-input entry 28:
$300000 \rightarrow 0$. & 1.437 & $-0.134$ & $-0.007$ \\

Least-loss seller expansion & The seller-side guide at bus 22 is increased from $0.343$ MW to $0.700$ MW. & Seller guide at bus 22, public-input entry 15: $343489 \rightarrow 700000$. & 1.601 & 0.031 & 0.004 \\

Sensitivity overestimation & The guide is recomputed after overestimating all voltage and line-flow sensitivities by $20\%$, producing a more conservative guide. & Multiple guide entries are modified, including public-input entries 15 and 30. &
1.487 & $-0.083$ & $-0.008$ \\
\hline
\end{tabular*}
\end{table*}
\subsection{Transaction Guide Performance}
Fig.~\ref{fig:volandline} compares the network states obtained by applying an AC power flow to the cleared trading results. The base case causes severe violations of both the lower voltage limit and the thermal limit. The universal policy reduces the line loading below the thermal limit, but its minimum voltage still falls below the admissible bound. The causality-based policy also violates both the voltage and thermal constraints under the AC power-flow evaluation. In contrast, the transaction-guide policy is the only case that remains feasible, keeping the realized operating point inside both the voltage and line-loading limits.
\begin{table}[t]
    \centering
    \caption{Market and network outcomes estimated by AC power flow under different policies}
    \label{tab:market-performance}
    \begin{tabular}{lcccc}
        \hline
        Policy & \makecell{Trade\\$[\mathrm{MW}]$} & \makecell{Min. voltage \\$[\mathrm{p.u.}]$} & \makecell{Max. loading \\$[\%]$} & \makecell{Feasibility} \\
        \hline
        Base case         & 3.638 & 0.859 & 141.506 & No  \\
        Universal         & 0.913 & 0.946 & 95.237  & No  \\
        Causality-based   & 2.152 & 0.946 & 102.758 & No  \\
        Transaction guide & 1.570 & 0.955 & 93.514  & Yes \\
        \hline
    \end{tabular}\vspace{-0mm}
\end{table}
Table~\ref{tab:market-performance} explains these profiles through the trading results. The base case clears 3.638 MW, causing the minimum voltage to drop to 0.859 p.u. and the maximum line loading to reach 141.506\%. The universal policy substantially curtails trade to 0.913 MW and reduces the maximum line loading to 95.237\%, but the minimum voltage remains below 0.95 p.u. limit. The causality-based policy clears 2.152 MW; however, the resulting AC operating point remains infeasible, with a minimum voltage of 0.946 p.u. and a maximum line loading of 102.758\%. Although the universal and causality-based policies reduce trading volumes using network-security signals—uniform violation penalties in the former and sensitivity-based charges in the latter—their corrective actions are still affected by linearization errors, leading to constraint violations under the nonlinear AC power-flow evaluation. By contrast, the transaction-guide policy clears 1.570 MW while satisfying both voltage and line-flow constraints, yielding a minimum voltage of 0.955 p.u. and a maximum line loading of 93.514\%. This is because the proposed transaction guide imposes an ex-ante robust envelope on admissible transactions, thereby limiting the cleared volume to a level that does not induce network violations under the AC operating condition.

\begin{table}[t]
    \centering
    \caption{Representative bus-wise transaction-guide values and realized transactions}
    \vspace{-2mm}
    \label{tab:buswise-guide}
    \begin{tabular}{cccccc}
        \hline
        Bus & \makecell{Seller cap\\$[\mathrm{MW}]$} & \makecell{Buyer cap\\$[\mathrm{MW}]$} 
            & \makecell{$u_i$\\$[\mathrm{MW}]$} & \makecell{$\ell_i$\\$[\mathrm{MW}]$} 
            & \makecell{Realized trade \\$[\mathrm{MW}]$} \\
        \hline
        22 & 0.700 & 0     & 0.343 & 0     & 0.340 \\
        23 & 0.700 & 0     & 0.700 & 0     & 0.406 \\
        30 & 0.750 & 0     & 0.750 & 0     & 0.408 \\
        31 & 0.650 & 0     & 0.650 & 0     & 0.417 \\
        12 & 0     & 0.250 & 0     & 0.113 & 0.074 \\
        32 & 0     & 0.200 & 0     & 0.200 & 0.130 \\
        \hline
    \end{tabular}\vspace{-0mm}
\end{table}

Table~\ref{tab:buswise-guide} reports representative guide values and realized transactions. For sellers, realized injections remain below the upward guides; for buyers, realized volume remains below the downward guides. The guide values show how network constraints shape admissible transactions: bus 22 receives only 0.343 MW despite a seller cap of 0.700 MW, whereas buses 23, 30, and 31 receive their full seller caps. Thus, the proposed method publishes a verifiable envelope before market clearing and, under the subsequent AC power-flow evaluation, prevents transactions that would drive the network outside the admissible operating region.

\subsection{Verification Statement and Tampering Scenarios}
We evaluate whether zero-knowledge verification detects utility-side manipulation of the certified guide. The attacks modify the published guide after proof generation while reusing the original proof. Peers observe only the verification key, the public input vector, and the proof $(A,B,C)$; any inconsistency must be detected through the pairing equation.

Table~\ref{tab:utility_tamper_scenarios} summarizes the tampering scenarios. Market quantities are rounded to three decimal places, and public-input values use the fixed-point scaling factor $10^6$. Public-input entry numbers refer to the serialized Groth16 input vector, not physical bus numbers. Changes in trade and network loss are measured relative to the honest transaction-guide case.

The first two scenarios restrict selected seller- or buyer-side guides and reduce cleared P2P trade by $0.386$ MW and $0.134$ MW, respectively. These changes can increase the utility's fallback retail sales proxy, with different loss impacts. The third scenario selectively expands the seller-side guide at bus 22 from $0.343$ MW to $0.700$ MW, preferentially increasing one seller's admissible width while increasing loss by $0.004$ MW. The fourth scenario overestimates voltage and line-flow sensitivities by $20\%$, producing a more conservative guide and reducing cleared trade by $0.083$ MW.

Table~\ref{tab:pairing_tamper_results} reports the resulting pairing outcomes. For compact reporting, the left and right-hand sides of \eqref{eq:verifeq} are shown as hexadecimal digests of the serialized target-group elements. The honest proof is accepted. In all tampering scenarios, the public input changes while the original proof is reused; the reconstructed $vk_x$ changes, \eqref{eq:verifeq} fails, and the proof is rejected.

\begin{table}[t]
\centering
\caption{Verification results under transaction-guide tampering}
\label{tab:pairing_tamper_results}
\footnotesize
\renewcommand{\arraystretch}{1.15}
    \begin{tabular}{
    p{0.35\linewidth}
    >{\centering\arraybackslash}p{0.20\linewidth}
    >{\centering\arraybackslash}p{0.20\linewidth}
    >{\centering\arraybackslash}p{0.09\linewidth}}
        \hline
        Scenario & LHS value & RHS value & Result \\
        \hline
        Honest TG & \texttt{a76e9c9c3e6e} & \texttt{a76e9c9c3e6e} & Pass \\
        Low-cost seller restriction & \texttt{c9ad92401b16} & \texttt{a76e9c9c3e6e} & Fail \\
        Buyer guide restriction & \texttt{cddaea144df6} & \texttt{a76e9c9c3e6e} & Fail \\ 
        Least-loss seller expansion & \texttt{43e57b4cf08f} & \texttt{a76e9c9c3e6e} & Fail \\
        Sensitivity overestimation & \texttt{bb79b8332608}  & \texttt{a76e9c9c3e6e} & Fail \\
        \hline
    \end{tabular}
\end{table}
These results show that the certified guide is cryptographically bound to the Groth16 proof through the public input vector. Tampered guides may change market outcomes, but they are inconsistent with the proof generated for the honest guide and are rejected by the verifier.

\subsection{Zero-Knowledge Proof on Blockchain}
This subsection evaluates the computational and blockchain overhead. The ZKP circuit proves guide-cap constraints, guide-balance, network-security inequalities, and the binding relation between private network parameters and the public verification statement. For the IEEE 33-bus test feeder, the certificate includes 24 participant buses, 48 guide variables, and 96 security rows. The resulting Groth16 circuit contains 76,565 R1CS constraints and 82,885 wires.

\begin{table}[t]
    \centering
    \caption{ZKP circuit and artifact overhead}
    \label{tab:zkp_overhead}
    \setlength{\tabcolsep}{4pt}
    \renewcommand{\arraystretch}{1.10}
    \begin{tabularx}{\linewidth}{
        >{\raggedright\arraybackslash}X
        r
    }
    \hline
    Metric & Value \\
    \hline
    Participant buses & 24 \\
    Guide variables & 48 \\
    Security rows & 96 \\
    Total injection guide & 2.443 MW \\
    Total withdrawal guide & 2.443 MW \\
    R1CS constraints & 76,565 \\
    Wires & 82,885 \\
    Public inputs & 194 \\
    Proof artifact size & 806 bytes \\
    Public-input JSON size & 1,889 bytes \\
    Verification-key size & 38,252 bytes \\
    input size & 2,652,396 bytes \\
    R1CS file size & 27,120,948 bytes \\
    \texttt{snarkjs} verification & Passed \\
    \hline
    \end{tabularx}
\end{table}

Table~\ref{tab:zkp_overhead} summarizes the circuit size and artifacts. The proof is 806 bytes because a Groth16 proof has a constant number of group elements. The verification key, input, and R1CS artifacts are larger because they encode the circuit structure or complete assignment, but they are not repeatedly submitted during normal operation.

To evaluate blockchain overhead, the Groth16 verifier and proof-registry contract are deployed in a local Hardhat EVM environment. The verifier takes $(A,B,C)$ and 194 public inputs, while the registry stores the circuit identifier, verification-key hash, public-input hash, proof hash, and verification result.

\begin{table}[t]
    \centering
    \caption{Blockchain verification overhead and indicative USD cost}\vspace{0mm}
    \label{tab:blockchain_overhead}
    \setlength{\tabcolsep}{3pt}
    \renewcommand{\arraystretch}{1.15}
    \begin{tabular}{p{0.50\linewidth}rr}
    \hline
    Metric & Value & Cost [USD] \\
    \hline
    Public-input count & 194 & -- \\
    Direct Solidity verifier result & True & -- \\
    Registry static-call result & True & -- \\
    Stored verification result & True & -- \\
    Transaction status & 1 & -- \\
    Circuit registration gas & 113,628 & 0.011 \\
    Proof submission gas & 1,761,007 & 0.169 \\
    Total gas & 1,874,635 & 0.180 \\
    \hline
    \end{tabular}
    \vspace{0mm}
    \begin{minipage}{0.95\linewidth}
    \footnotesize
    Cost estimates use a gas price of $0.059$ gwei per gas. One gwei is equal to $10^{-9}$ ETH~\cite{ethereum_gas_docs}, 
    and the ETH price is fixed at $1{,}625.80$ USD/ETH based on the market price at the time of access~\cite{eth_price_cmc}.
    \end{minipage}
\end{table}
Table~\ref{tab:blockchain_overhead} reports measured gas. Circuit registration requires 113,628 gas, while proof submission and verification require 1,761,007 gas, yielding a total of 1,874,635 gas. The proof-submission overhead is dominated by BN254 pairing operations and calldata/storage for the 194 public inputs.

\section{Conclusion}
This paper proposed a zero-knowledge verification method for network-constrained transaction guides for P2P energy trading. The method addresses the confidentiality–verifiability tradeoff in sensitivity-based coordination by enabling participants to verify the computational integrity of utility-provided guides without revealing private network sensitivities. A robust guide is formulated using sign-decomposed voltage and line-flow sensitivities, and its feasibility and LP optimality were encoded as R1CS/QAP constraints and verified through a Groth16 proof.

Blockchain anchoring of circuit records, statement commitments, public inputs, verification keys, nonces, and proof results provides tamper-evident auditability while keeping confidential data off-chain. Case studies on a modified IEEE 33-bus distribution system showed that the proposed guide maintains post-clearing network feasibility and rejects representative public-input and witness-inconsistency attacks. The implemented circuit produced an 806-byte proof and practical on-chain verification overhead in the tested case. These results indicate that zero-knowledge verification can serve as a computational trust layer for distribution-level energy markets. 
\appendix
\subsection{Trigonometric Approximation in R1CS}\label{App:trig}
For line $(i,j)\in\mathcal L$, define
\begin{equation}
    \theta_{ij}-\theta_i^{0}+\theta_j^{0} = 0.
    \label{eq:r1cs_theta_ij_lin}
\end{equation}
We use truncated approximations
\begin{equation}
    s_{ij}
    =
    \theta_{ij}-\frac{\theta_{ij}^{3}}{6}+\frac{\theta_{ij}^{5}}{120},
    \qquad
    c_{ij}
    =
    1-\frac{\theta_{ij}^{2}}{2}+\frac{\theta_{ij}^{4}}{24}.
    \label{eq:approx_trig}
\end{equation}
With $\hat\theta_{ij}=\theta_{ij}^{2}$, these approximations are enforced by
\begin{equation}
    \label{eq:r1cs_trig_edge_min}
    \begin{alignedat}{2}
    \theta_{ij}\cdot\theta_{ij} &= \hat\theta_{ij},
    &\quad \hat\theta_{ij}\cdot 120^{-1}-6^{-1} &= t^{\mathrm s}_{ij},\\
    \hat\theta_{ij}\cdot t^{\mathrm s}_{ij} &= u^{\mathrm s}_{ij},
    &\quad u^{\mathrm s}_{ij}+1 &= v^{\mathrm s}_{ij},\\
    \theta_{ij}\cdot v^{\mathrm s}_{ij} &= s_{ij},
    &\quad \hat\theta_{ij}\cdot 24^{-1}-2^{-1} &= t^{\mathrm c}_{ij},\\
    \hat\theta_{ij}\cdot t^{\mathrm c}_{ij} &= u^{\mathrm c}_{ij},
    &\quad u^{\mathrm c}_{ij}+1 &= c_{ij}.
    \end{alignedat}
\end{equation}
where $2^{-1},6^{-1},24^{-1},120^{-1}\in\mathbb F_r$ are fixed constants.
The approximation error is bounded over the operating angle interval \(|\theta_{ij}|\le \Theta_{\max}\). For the fifth-order sine and fourth-order cosine approximations,
\[
|\sin\theta_{ij}-s_{ij}| \le \frac{\Theta_{\max}^7}{7!},\quad
|\cos\theta_{ij}-c_{ij}| \le \frac{\Theta_{\max}^6}{6!}.
\]
These bounds are propagated through the Jacobian and line-flow sensitivity expressions to obtain the tightening margins \(\epsilon_V^{\mathrm{trig}}\) and \(\epsilon_S^{\mathrm{trig}}\). In the case study, \(\Theta_{\max}\) is computed from the base AC power-flow solution and the resulting approximation margins are included in the circuit security margins.

\subsection{Line-Flow in R1CS}\label{App:flowEquation}
For each monitored line $l=(i,j)\in\mathcal L$, the line flow model without explicit line-charging terms is enforced by
\begin{equation}
    \label{eq:r1cs_flow_using_Yoff}
    \begin{alignedat}{2}
    d^{\mathrm p}_{ij}+\alpha_i G_{ij} &= 0,
    &\quad
    \alpha_iB_{ij}-d^{\mathrm q}_{ij} &= 0,\\
    \beta_{ij}v_{ij}-w^{\mathrm p}_{ij} &= 0,
    &\quad
    \beta_{ij}u_{ij}-w^{\mathrm q}_{ij} &= 0,\\
    d^{\mathrm p}_{ij}+w^{\mathrm p}_{ij}-P^{\mathrm F}_{ij} &= 0,
    &\quad
    d^{\mathrm q}_{ij}+w^{\mathrm q}_{ij}-Q^{\mathrm F}_{ij} &= 0.
    \end{alignedat}
\end{equation}
Thus, $P^{\mathrm F}_{ij}=-\alpha_iG_{ij}+\beta_{ij}v_{ij}$ and $Q^{\mathrm F}_{ij}=\alpha_iB_{ij}+\beta_{ij}u_{ij}$. Auxiliary products bind $|S^{\mathrm F,0}_{ij}|$:
    \begin{align}\label{eq:r1cs_Smag_bind_using_Yoff}
        \begin{split}
        P^{\mathrm F}_{ij}\cdot P^{\mathrm F}_{ij} &= P^{\mathrm F,2}_{ij},\\
        Q^{\mathrm F}_{ij}\cdot Q^{\mathrm F}_{ij} &= Q^{\mathrm F,2}_{ij},\\
        P^{\mathrm F,2}_{ij}+Q^{\mathrm F,2}_{ij}-(|S^{\mathrm F,0}_{ij}|)^2 &= 0.
        \end{split}
    \end{align}
\subsection{Line-Flow Sensitivity in R1CS}\label{App:lineflowsens}
The partial derivatives of $P^{\mathrm F}_{ij}$ and $Q^{\mathrm F}_{ij}$ are
\begin{equation}
    \label{eq:partial_lpf}
    \begin{gathered}
    \begin{alignedat}{2}
    \beta_{ij}u_{ij} +\frac{\partial P^{\mathrm F}_{ij}}{\partial \theta_i} &= 0, &\quad \beta_{ij}u_{ij} -\frac{\partial P^{\mathrm F}_{ij}}{\partial \theta_j} &= 0,\\
    \beta_{ij}v_{ij} -\frac{\partial Q^{\mathrm F}_{ij}}{\partial \theta_i} &= 0, &\quad \beta_{ij}v_{ij} +\frac{\partial Q^{\mathrm F}_{ij}}{\partial \theta_j} &= 0,\\
    |V_i^0|u_{ij} -\frac{\partial Q^{\mathrm F}_{ij}}{\partial |V_j|} &= 0, &\quad |V_i^0|v_{ij} -\frac{\partial P^{\mathrm F}_{ij}}{\partial |V_j|} &= 0,\\
    \end{alignedat}
    \\[-0.2em]
    \begin{aligned}
    -2|V_i^0|G_{ij}
    +|V_j^0|v_{ij}
    -\frac{\partial P^{\mathrm F}_{ij}}{\partial |V_i|} &= 0,\\
    2|V_i^0|B_{ij}
    +|V_j^0|u_{ij}
    -\frac{\partial Q^{\mathrm F}_{ij}}{\partial |V_i|} &= 0.
    \end{aligned}
    \end{gathered}
\end{equation}
Using \eqref{eq:partial_lpf}, auxiliary products decompose $\Sigma^{P}_{l,k}$ and $\Sigma^{Q}_{l,k}$:
\begin{align}\label{eq:r1cs_S_sensitivity_terms}
    \begin{split}
    &\frac{\partial P^{\mathrm F}_{l}}{\partial |V_r|} \cdot A^{V}_{rk} = t^{P,V_r}_{l,k}, \; \frac{\partial Q^{\mathrm F}_{l}}{\partial |V_r|} \cdot A^{V}_{rk} = t^{Q,V_r}_{l,k}, \; r\in\{i,j\},\\
    &\frac{\partial P^{\mathrm F}_{l}}{\partial \theta_r} \cdot A^{\theta}_{rk} = t^{P,\theta_r}_{l,k}, \; \frac{\partial Q^{\mathrm F}_{l}}{\partial \theta_r} \cdot A^{\theta}_{rk} =t^{Q,\theta_r}_{l,k},
    \; r\in\{i,j\},\\
    &t^{P,V_i}_{l,k} +t^{P,V_j}_{l,k} +t^{P,\theta_i}_{l,k} +t^{P,\theta_j}_{l,k} = \Sigma^{P}_{l,k},\\
    &t^{Q,V_i}_{l,k} +t^{Q,V_j}_{l,k} +t^{Q,\theta_i}_{l,k} +t^{Q,\theta_j}_{l,k} = \Sigma^{Q}_{l,k}.
    \end{split}
\end{align}

\bibliographystyle{IEEEtran}
\bibliography{reference}
\end{document}